\documentclass[11pt]{article}

\def\colorful{0}

\usepackage{amsthm,amsfonts,amsmath,amssymb,epsfig,color,float,graphicx,verbatim}
\usepackage{multirow}

\newif\ifhyper\IfFileExists{hyperref.sty}{\hypertrue}{\hyperfalse}
\hypertrue
\ifhyper\usepackage{hyperref}\fi

\usepackage{enumitem}
\usepackage[capitalize]{cleveref} 

\makeatletter
\renewcommand{\section}{\@startsection{section}{1}{0pt}{-12pt}{5pt}{\large\bf}}
\renewcommand{\subsection}{\@startsection{subsection}{2}{0pt}{-12pt}{-5pt}{\normalsize\bf}}
\renewcommand{\subsubsection}{\@startsection{subsubsection}{3}{0pt}{-12pt}{-5pt}{\normalsize\bf}}
\makeatother

\usepackage{framed}
\usepackage{nicefrac}

\def\nnewcolor{1}
\ifnum\nnewcolor=1

\fi
\ifnum\nnewcolor=0

\fi

\ifnum\colorful=1
\newcommand{\new}[1]{{\color{red} #1}}

\else
\newcommand{\new}[1]{{#1}}

\fi

\newtheorem{theorem}{Theorem}

\newtheorem{lemma}[theorem]{Lemma}
\newtheorem{proposition}[theorem]{Proposition}
\newtheorem{corollary}[theorem]{Corollary}
\newtheorem{claim}[theorem]{Claim}

\theoremstyle{definition}

\newcommand{\R}{\mathbb{R}}
\newcommand{\C}{\mathbb{C}}
\newcommand{\Z}{\mathbb{Z}}
\newcommand{\E}{\mathbb{E}}

\newcommand{\poly}{\mathrm{poly}}
\newcommand{\polylog}{\mathrm{polylog}}

\newcommand{\p}{\mathbf{P}}
\newcommand{\q}{\mathbf{Q}}
\newcommand{\h}{\mathbf{H}}

\newcommand{\dtv}{d_{\mathrm TV}}

\newcommand{\var}{\mathrm{Var}}

\newcommand{\ignore}[1]{}

\newcommand{\eps}{\epsilon}

\renewcommand{\eqref}[1]{Eq.~(\ref{#1})}

\newcommand{\eqdef}{\stackrel{{\mathrm {\footnotesize def}}}{=}}

\newenvironment{algorithm}[1][\  ] %
{ \rm
\begin{tabbing}
....\=.....\=.....\=.....\=.....\=  \+ \kill
} %
{\end{tabbing} }

\floatstyle{ruled}
\newfloat{Algorithm}{H}{alg}
\newfloat{Subroutine}{H}{sub}

{
\begin{minipage}{1.0\linewidth} \begin{algorithm} %
} { \end{algorithm} \end{minipage} }

\title{Properly Learning Poisson Binomial Distributions\\ in Almost Polynomial Time}

\author{
Ilias Diakonikolas\thanks{Supported by EPSRC grant EP/L021749/1 and a Marie Curie Career Integration grant.}\\
University of Edinburgh\\
{\tt ilias.d@ed.ac.uk}.\\
\and
Daniel M. Kane\thanks{Some of this work was performed while visiting the University of Edinburgh.}\\
University of California, San Diego\\
{\tt dakane@cs.ucsd.edu}.\\
\and
Alistair Stewart\thanks{Supported by EPSRC grant EP/L021749/1.}\\
University of Edinburgh\\
{\tt stewart.al@gmail.com}.
}

\begin{document}

\maketitle

\thispagestyle{empty}

\begin{abstract}
We give an algorithm for properly learning Poisson binomial distributions.
A Poisson binomial distribution (PBD) of order $n \in \Z_+$
is the discrete probability distribution of the sum of $n$ mutually independent Bernoulli random variables.
Given $\widetilde{O}(1/\eps^2)$ samples from an unknown PBD $\p$, our algorithm runs in time
$(1/\eps)^{O(\log \log (1/\eps))}$, and outputs a hypothesis PBD that is $\eps$-close to $\p$ in total variation distance.
The sample complexity of our algorithm is known to be nearly-optimal, up to logarithmic factors, as established
in previous work~\cite{DDS12stoc}. However, the previously best known running time for properly
learning PBDs~\cite{DDS12stoc, DKS15} was $(1/\eps)^{O(\log(1/\eps))}$, and was essentially obtained by
enumeration over an appropriate  $\eps$-cover. We remark that the running time of this cover-based approach cannot be
improved, as any $\eps$-cover for the space of PBDs has size  $(1/\eps)^{\Omega(\log(1/\eps))}$~\cite{DKS15}.

As one of our main contributions, we provide a novel structural characterization of PBDs,
showing that any PBD $\p$ is $\eps$-close to another PBD $\q$ with $O(\log(1/\eps))$ distinct parameters.
More precisely, we prove that, for all $\eps >0,$ there exists
an explicit collection $\cal{M}$ of $(1/\eps)^{O(\log \log (1/\eps))}$ vectors of multiplicities,
such that for any PBD $\p$ there exists a PBD $\q$ with $O(\log(1/\eps))$ distinct parameters whose multiplicities are given by some element of ${\cal M}$,
such that $\q$ is $\eps$-close to $\p.$  Our proof combines tools from Fourier analysis and algebraic geometry.

Our approach to the proper learning problem is as follows:
Starting with an accurate non-proper hypothesis, we fit a PBD to this hypothesis.
More specifically, we essentially start with the hypothesis computed by the
computationally efficient non-proper learning algorithm in our recent work~\cite{DKS15}.
Our aforementioned structural characterization allows
us to reduce the corresponding fitting problem
to a collection of $(1/\eps)^{O(\log \log(1/\eps))}$ 
systems of low-degree polynomial inequalities.
We show that each such system can be solved in time $(1/\eps)^{O(\log \log(1/\eps))}$,
which yields the overall running time of our algorithm.
\end{abstract}

\thispagestyle{empty}
\setcounter{page}{0}

\newpage

\section{Introduction}  \label{sec:intro}

The Poisson binomial distribution (PBD) is the discrete probability distribution
of a sum of mutually independent Bernoulli random variables.
PBDs comprise one of the most fundamental nonparametric families of discrete distributions.
They have been extensively studied in probability and statistics~\cite{Poisson:37, Chernoff:52,Hoeffding:63,DP09},
and are ubiquitous in various applications (see, e.g., ~\cite{ChenLiu:97} and references therein). Recent years have witnessed a flurry of research activity
on PBDs and generalizations from several perspectives of theoretical computer science,
including learning~\cite{DDS12stoc, DDOST13focs, DKS15, DKT15, DKS15b}, pseudorandomness and derandomization~\cite{GMRZ11, BDS12, De15, GKM15},
property testing~\cite{AD15, CDGR15},
and computational game theory~\cite{DaskalakisP07, DaskalakisP09, DP:cover, DaskalakisP2014, GT14}.

Despite their seeming simplicity, PBDs have surprisingly rich structure, and basic questions about them can be
unexpectedly challenging to answer. We cannot do justice to the probability literature studying the following
question: Under what conditions can we approximate PBDs by simpler distributions?
See Section~1.2 of~\cite{DDS15-journal} for a summary.
In recent years, a number of works in theoretical computer science~\cite{DaskalakisP07, DaskalakisP09, DDS12stoc, DP:cover, DKS15}
have studied, and essentially resolved, the following questions:
Is there a small set of distributions that approximately cover the set of all PBDs? What is the
number of samples required to learn an unknown PBD?

We study the following natural computational question: Given independent samples from an unknown PBD $\p$, can we efficiently
find a hypothesis PBD $\q$ that is close to $\p$, in total variation distance?
That is, we are interested in {\em properly learning} PBDs, a problem
that has resisted recent efforts~\cite{DDS12stoc, DKS15} at designing efficient algorithms.
In this work, we propose a new approach to this problem that
leads to a significantly faster algorithm than was previously known. At a high-level, we establish
an interesting connection of this problem to algebraic geometry and polynomial optimization.
By building on this connection, we provide a new structural characterization
of the space of PBDs, on which our algorithm relies, that we believe is of independent interest.
In the following, we motivate and describe our results in detail, and elaborate on our ideas and techniques.

\medskip

\noindent {\bf Distribution Learning.} We recall the standard definition of learning
 an unknown probability distribution from samples~\cite{KMR+:94short, DL:01}:
Given access to independent samples drawn from an unknown distribution $\p$
in a given family ${\cal C}$, and an error parameter $\eps>0$, a learning algorithm for ${\cal C}$
must output a hypothesis $\h$ such that, with probability at least $9/10$,
the total variation distance 
between $\h$ and $\p$ is at most~$\eps$.
The performance of a learning algorithm is measured by its
{\em sample complexity} (the number of samples drawn from $\p$) and its {\em computational complexity}.

In {\em non-proper} learning (density estimation), the goal is to output
an approximation to the target distribution without any constraints on its representation.
In {\em proper} learning, we require in addition that
 the hypothesis $\h$ is a member of the family ${\cal C}$.
Note that these two notions of learning are essentially equivalent in terms of sample complexity
(given any accurate hypothesis, we can do a brute-force search to find its closest distribution in ${\cal C}$),
but not necessarily equivalent in terms of computational complexity.
A typically more demanding notion of learning is
that of {\em parameter estimation}. The goal here is to identify the parameters
of the unknown model, e.g., the means of the individual Bernoulli components for the case of PBDs,
up to a desired accuracy $\eps$.

\medskip

\noindent {\bf Discussion.}
In many learning situations, it is desirable to compute a proper hypothesis,
i.e., one that belongs to the underlying distribution family ${\cal C}$. A proper hypothesis is typically
preferable due to its interpretability. In the context of distribution learning,
a practitioner may not want to use a density estimate, unless it is proper.
For example, one may want the estimate to have the properties of the underlying family,
either because this reflects some physical understanding of the inference problem,
or because one might only be using the density estimate as the first stage of a more involved procedure.
While parameter estimation may arguably provide a more desirable guarantee
than proper learning in some cases, its sample complexity is typically prohibitively
large.

For the class of PBDs, we show (Proposition~\ref{prop:param}, Appendix~\ref{sec:param}) that parameter estimation
requires {$2^{\Omega(1/\eps)}$} samples, for PBDs with $n = \Omega(1/\eps)$ Bernoulli
components, where $\eps>0$ is the accuracy parameter.
In contrast, the sample complexity of (non-)proper learning is known to be $\widetilde{O}(1/\eps^2)$~\cite{DDS12stoc}.
Hence, proper learning serves as an attractive middle ground between non-proper learning and parameter estimation.
Ideally, one could obtain a proper learner for a given family
whose running time matches that of the best non-proper algorithm.

Recent work by the authors~\cite{DKS15}
has characterized the computational complexity
of non-properly learning PBDs, which was shown to be~$\widetilde{O}(1/\eps^2)$, i.e.,
nearly-linear in the sample complexity of the problem.
Motivated by this progress,  a natural research direction
is to obtain a computationally efficient proper learning algorithm,
i.e., one that runs in time $\poly(1/\eps)$ and outputs a PBD as its hypothesis.
Besides practical applications,
we feel that this is an interesting algorithmic problem, with intriguing connections to algebraic geometry
and polynomial optimization (as we point out in this work). We remark that several natural approaches
fall short of yielding a polynomial--time algorithm. More specifically, the proper learning of PBDs
can be phrased in a number of ways
as a structured non-convex optimization problem, albeit it is
unclear whether any such formulation may lead to a polynomial--time algorithm.

This work is part of a broader agenda of systematically investigating
the computational complexity of proper distribution learning. We believe that this
is a fundamental goal that warrants study for its own sake. The complexity of proper learning
has been extensively investigated in the supervised setting of PAC learning
Boolean functions~\cite{KearnsVazirani:94, Feldman15}, with several algorithmic and computational
intractability results obtained in the past couple of decades.  In sharp contrast, very little is known about the complexity of proper learning
in the unsupervised setting of learning probability distributions.

\subsection{Preliminaries.} \label{sec:prelim}
For $n, m \in \Z_+$ with $m \le n$, we will denote $[n] \eqdef \{0,1,\dots,n\}$ and $[m, n]  \eqdef \{m, m+1 ,\dots,n\}$.
For a distribution $\p$ supported on $[m]$, $m \in \Z_+$, we write $\p(i)$ to denote
the value $\Pr_{X \sim \p}[X=i]$ of the probability mass function (pmf) at point $i$.
The {\em total variation distance} between two distributions
$\p$ and $\q$ supported on a finite domain $A$ is
$\dtv\left(\p, \q \right) \eqdef \max_{S \subseteq A} \left |\p(S)-\q(S) \right|= (1/2) \cdot \| \p -\q  \|_1.$
If $X$ and $Y$ are random variables, their total
variation distance $\dtv(X,Y)$ is defined as the total variation
distance between their distributions.

\smallskip

\noindent {\bf Poisson Binomial Distribution.}
A {\em Poisson binomial distribution of order $n \in \mathbb{Z}_+$} or {\em $n$-PBD}
is the discrete probability distribution of the sum $\sum_{i=1}^n X_i$ of $n$
mutually independent Bernoulli random variables $X_1,\ldots,X_n$. 
An $n$-PBD $\p$ can be represented uniquely as the vector of its $n$ parameters $p_1, \ldots, p_n$,
i.e., as $(p_i)_{i=1}^n$, where we can assume that
$0\le p_1 \le p_2 \le \ldots \le p_n \le 1$. To go from $\p$ to its corresponding vector,
we find a collection $X_1,\ldots,X_n$ of mutually independent Bernoullis
such that $\sum_{i=1}^n X_i$ is distributed according to $\p$
with $\E[X_1]\le \ldots \le \E[X_n]$, and we set $p_i = \E[X_i]$ for all $i$.
\new{An equivalent unique representation of an $n$-PBD
with parameter vector  $(p_i)_{i=1}^n$ is via the vector of its {\em distinct  parameters}
 $p'_1, \ldots, p'_k$, where $1\le k \le n$, and $p'_i \ne p'_j$ for $i \neq j$,
 together with their corresponding integer multiplicities $m_1, \ldots, m_k$. Note that $m_i \ge 1$, $1 \le i \le k$, and $\sum_{i=1}^k m_i = n$.
 This representation will be crucial for the results and techniques of this paper.}

\smallskip

\noindent {\bf Discrete Fourier Transform.} For $x \in \R$ we will denote $e(x) \eqdef  \exp(2 \pi i x)$.
The {\em Discrete Fourier Transform (DFT) modulo $M$} of a function
$F:[n] \rightarrow \C$ is  the function $\widehat{F}:[M-1]\rightarrow \C$ defined as
$\widehat{F}(\xi)=\sum_{j=0}^n e(-\xi j/M) F(j) \;,$
for integers $\xi \in [M-1]$. The DFT modulo $M$, $\widehat{\p}$, of a distribution $\p$ is the DFT modulo $M$ of its probability mass function.
The {\em inverse DFT modulo $M$} onto the range $[m,m+M-1]$ of
$\widehat{F}: [M-1] \rightarrow \C$,
is the function $F: [m, m+M-1] \cap \Z \rightarrow \C$ defined by
$F(j)= \frac{1}{M} \sum_{\xi=0}^{M-1} e(\xi j/M) \widehat{F}(\xi) \;,$
for $j \in [m,m+M-1] \cap \mathbb{Z}$.
The $L_2$ norm of the DFT is defined as $\|\widehat{F}\|_2 = \sqrt{\frac{1}{M} \sum_{\xi=0}^{M-1} |\widehat{F}(\xi)|^2} \;.$

\subsection{Our Results and Comparison to Prior Work.}  \label{ssec:results}

We are ready to formally describe the main contributions of this paper.
As our main algorithmic result, we obtain a near-sample optimal and
almost polynomial-time algorithm for properly learning PBDs:

\begin{theorem} [Proper Learning of PBDs]  \label{thm:proper-learning}
For all $n \in \Z_+$ and $\eps>0$, there is a proper learning algorithm for $n$-PBDs with the following
performance guarantee: Let $\p$ be an unknown $n$-PBD.
The algorithm uses $\widetilde{O}(1/\eps^2)$ samples from $\p$,
runs in time $(1/\eps)^{O(\log \log (1/\eps))}$\footnote{We work
in the standard ``word RAM'' model in which basic arithmetic
operations on $O(\log n)$-bit integers are assumed to take constant time.},
and outputs (a succinct description of) an $n$-PBD $\q$ such that with probability at least $9/10$ it holds that $\dtv(\q, \p) \le \eps.$
\end{theorem}

We now provide a comparison of Theorem~\ref{thm:proper-learning} to previous work.
The problem of learning PBDs was first explicitly considered by Daskalakis {\em et al.}~\cite{DDS12stoc}, who
gave two main results: (i) a non-proper learning algorithm with sample complexity and running time
$\widetilde{O}(1/\eps^3)$, and (ii) a proper learning algorithm with sample complexity
$\widetilde{O}(1/\eps^2)$ and running time $(1/\eps)^{\polylog(1/\eps)}$.
In recent work~\cite{DKS15}, the authors of the current paper obtained a near-optimal sample and time
algorithm to non-properly learn a more general family of discrete distributions (containing PBDs).
For the special case of PBDs, the aforementioned work~\cite{DKS15} yields the following implications: (i) a non-proper learning algorithm with sample
and time complexity $\widetilde{O}(1/\eps^2)$, and (ii) a proper learning algorithm with sample complexity
$\widetilde{O}(1/\eps^2)$ and running time $(1/\eps)^{\Theta(\log(1/\eps))}$.
 Prior to this paper, this was the fastest algorithm for properly learning PBDs.
 Hence, Theorem~\ref{thm:proper-learning} represents a super-polynomial improvement
in the running time, while still using a near-optimal sample size.

In addition to obtaining a significantly more efficient algorithm, the proof of
Theorem~\ref{thm:proper-learning} offers a novel approach to the problem of properly learning PBDs. The proper
algorithms of~\cite{DDS12stoc, DKS15} exploit the cover structure of the space of PBDs, and (essentially)
proceed by running an appropriate tournament procedure over an $\eps$-cover
(see, e.g., Lemma~10 in~\cite{DDS15-journal})\footnote{Note that any $\eps$-cover for the space of $n$-PBDs has size $\Omega(n)$. However, for the task of properly
learning PBDs, by a simple (known) reduction, one can assume without loss of generality that $n = \poly(1/\eps)$. Hence, the tournament-based
algorithm only needs to consider $\eps$-covers over PBDs with $\poly(1/\eps)$ Bernoulli components.}.
This cover-based approach, when applied to an $\eps$-covering set of size $N$, clearly has runtime $\Omega(N)$, and
can be easily implemented in time $O(N^2/\eps^2)$. \cite{DDS12stoc} applies the cover-based approach
to the $\eps$-cover construction of~\cite{DP:cover}, which has size $(1/\eps)^{O(\log^2(1/\eps))}$,
while ~\cite{DKS15} proves and uses a new cover construction of size $(1/\eps)^{O(\log(1/\eps))}$.
Observe that if there existed an explicit $\eps$-cover of size $\poly(1/\eps)$, the aforementioned cover-based approach
would immediately yield a $\poly(1/\eps)$ time proper learning algorithm. Perhaps surprisingly, it was shown in
~\cite{DKS15} that {\em any} $\eps$-cover for $n$-PBDs with $n = \Omega(\log(1/\eps))$ Bernoulli coordinates
has size $(1/\eps)^{\Omega(\log(1/\eps))}$. In conclusion, the cover-based approach for properly learning PBDs
inherently leads to runtime of $(1/\eps)^{\Omega(\log(1/\eps))}$. 

In this work, we circumvent the $(1/\eps)^{\Omega(\log(1/\eps))}$ cover size lower bound by establishing a new structural characterization
of the space of PBDs. Very roughly speaking, our structural result allows us to reduce the proper learning problem
to the case that the underlying PBD has $O(\log(1/\eps))$ {\em distinct} parameters.
Indeed, as a simple corollary of our main structural result (Theorem~\ref{thm:complex-distinct-param} in Section~\ref{sec:count}),
we obtain the following:

\begin{theorem} [A ``Few'' Distinct Parameters Suffice] \label{thm:simple-distinct-param}
For all $n \in \Z_+$ and $\eps>0$ the following holds:
For any $n$-PBD $\p$, there exists an $n$-PBD $\q$
with $\dtv(\p,\q) \leq \eps$ such that $\q$ has $O(\log(1/\eps))$ distinct parameters.
\end{theorem}

We note that in subsequent work~\cite{DKS15b} the authors generalize the above theorem
to Poisson multinomial distributions.

\medskip

\noindent {\bf Remark.} We remark that Theorem~\ref{thm:simple-distinct-param} is quantitatively tight, 
i.e., $O(\log(1/\eps))$ distinct parameters are in general necessary to $\eps$-approximate PBDs.
This follows directly from the explicit cover lower bound construction of~\cite{DKS15}.

\medskip

We view Theorem~\ref{thm:simple-distinct-param} as a natural structural result for PBDs.
Alas, its statement does not quite suffice for our algorithmic application.
While Theorem~\ref{thm:simple-distinct-param} guarantees that  $O(\log(1/\eps))$ distinct parameters are enough
to consider for an $\eps$-approximation,
it gives no information on the multiplicities these parameters may have.
In particular, the upper bound on the number of different combinations of multiplicities
one can derive from it is \new{$(1/\eps)^{O(\log(1/\eps))}$}, which
is not strong enough for our purposes.
The following stronger structural result (see Theorem~\ref{thm:complex-distinct-param}
and Lemma~\ref{NumberofMultiplicitiesLem} for detailed statements) is critical for our improved proper algorithm:

\begin{theorem}  [A ``Few'' Multiplicities and Distinct Parameters Suffice] \label{thm:complex-distinct-param-informal}
For all $n \in \Z_+$ and $\eps>0$ the following holds:
For any $\widetilde{\sigma}>0,$ there exists an explicit collection $\cal{M}$ of $(1/\eps)^{O(\log \log (1/\eps))}$ vectors of multiplicities computable in $\poly(|\cal{M}|)$ time,
so that for any $n$-PBD $\p$ with variance $\Theta(\widetilde{\sigma}^2)$
there exists a PBD $\q$ with $O(\log(1/\eps))$ distinct parameters whose multiplicities are given by some element of ${\cal M}$,
such that  $\dtv(\p,\q) \leq \eps$.
\end{theorem}

Now suppose we would like to properly learn an unknown PBD with $O(\log(1/\eps))$ {\em distinct} parameters and known multiplicities
for each parameter. Even for this very restricted subset of PBDs, the construction of~\cite{DKS15} implies
a cover lower bound of $(1/\eps)^{\Omega(\log(1/\eps))}$. To handle such PBDs,
we combine ingredients from Fourier analysis and algebraic geometry \new{with careful Taylor series approximations,}
to construct an appropriate system of low-degree polynomial inequalities
whose solution approximately recovers the unknown distinct parameters.


In the following subsection, we provide a detailed intuitive explanation of our techniques.

\subsection{Techniques.} \label{sec:techniques}

The starting point of this work lies in the non-proper learning algorithm from
our recent work~\cite{DKS15}. Roughly speaking, our new proper algorithm can be viewed as a two-step
process: We first compute an accurate non-proper hypothesis $\h$ using the algorithm in
~\cite{DKS15}, and we then post-process $\h$ to find  a PBD $\q$ that is close to $\h$.
We note that the non-proper hypothesis $\h$ output by~\cite{DKS15}
is represented succinctly via its Discrete Fourier Transform; this property
is crucial for the computational complexity of our proper algorithm.
(We note that the description of our proper algorithm and its analysis, presented in Section~\ref{sec:alg}, are entirely self-contained.
The above description is for the sake of the intuition.)

We now proceed to explain the connection in detail.
The crucial fact, established in~\cite{DKS15} for a more general setting, is that
the Fourier transform of a PBD has small effective support
(and in particular the effective support of the Fourier transform has size roughly inverse
to the effective support of the PBD itself). Hence,
in order to learn an unknown PBD $\p$, it suffices to find another PBD, $\q$,
with similar mean and standard deviation to $\p$,
so that the Fourier transform of $\q$
approximates the Fourier transform of $\p$
on this small region.
(The non-proper algorithm of ~\cite{DKS15} for PBDs
essentially outputs the empirical DFT of $\p$ over its effective support.)

Note that the Fourier transform of a PBD is the product of the Fourier transforms of its individual component variables.
By Taylor expanding the logarithm of the Fourier transform, we can write the log Fourier transform of a PBD
as a Taylor series whose coefficients are related to the moments of the parameters of $\p$ (see Equation (\ref{logTaylorEqn})).
We show that for our purposes it suffices to find a PBD $\q$ so that the first $O(\log(1/\epsilon))$ moments
of its parameters approximate the corresponding moments for $\p$.
Unfortunately, we do not actually know the moments for $\p$,
but since we can easily approximate the Fourier transform of $\p$ from samples,
we can derive conditions that are sufficient for the moments of $\q$ to satisfy.
This step essentially gives us a system of polynomial inequalities in the moments
of the parameters of $\q$ that we need to satisfy.

A standard way to solve such a polynomial system is by appealing to Renegar's algorithm~\cite{Renegar92-sicomp, Renegar92-jsc}, which
allows us to solve a system of degree-$d$ polynomial inequalities in $k$ real variables in time roughly $d^k$.
In our case, the degree $d$ will be at most poly-logarithmic in $1/\eps$, but the number of variables
$k$ corresponds to the number of parameters of $\q$, which is $k = \poly(1/\eps)$.
Hence, this approach is insufficient to obtain a faster proper algorithm.

To circumvent this obstacle, we show that it actually suffices to consider only PBDs with $O(\log(1/\epsilon))$ many distinct parameters
(Theorem~\ref{thm:simple-distinct-param}).
To prove this statement, we use a recent result from algebraic geometry due to Riener~\cite{R:11} (Theorem~\ref{thm:reiner}),
that can be used to relate the number of distinct parameters of a solution of a polynomial system
to the degree of the polynomials involved.
Note that the problem of matching $O(\log(1/\epsilon))$ moments
can be expressed as a system of polynomial equations,
where each polynomial has degree $O(\log(1/\epsilon))$. We can thus
find a PBD $\q$, which has the same first $O(\log(1/\epsilon))$ moments
as $\p$, with $O(\log(1/\epsilon))$ distinct parameters such that $\dtv (\q, \p) \le \eps.$
For PBDs with $O(\log(1/\eps))$ distinct parameters and {\em known} multiplicities for these parameters, we can reduce the runtime
of solving the polynomial system to $O(\log(1/\epsilon))^{O(\log(1/\epsilon))}=(1/\epsilon)^{O(\log\log(1/\epsilon))}.$

Unfortunately, the above structural result is not strong enough,
as in order to set up an appropriate system of polynomial inequalities for the parameters of $\q$,
we must first guess the multiplicities to which the distinct parameters appear.
A simple counting argument shows that there are roughly $k^{\log(1/\epsilon)}$ ways to choose these multiplicities.
To overcome this second obstacle, we need the following refinement of our structural result on distinct parameters:
We divide the parameters of $\q$ into categories based on how close they are to $0$ or $1$.
We show that there is a tradeoff between the number of parameters in a given category
and the number of distinct parameters in that category (see Theorem \ref{thm:complex-distinct-param}).
With this more refined result in hand, we show that there are only $(1/\epsilon)^{O(\log\log(1/\epsilon))}$
many possible collections of multiplicities that need to be considered (see Lemma \ref{NumberofMultiplicitiesLem}]).

Given this stronger structural characterization, our proper learning algorithm is fairly simple.
We enumerate over the set of possible collections of multiplicities as described above.
For each such collection, we set up a system of polynomial equations in the distinct parameters of $\q$,
so that solutions to the system will correspond to PBDs whose distinct parameters have the specified multiplicities
which are also $\eps$-close to $\p$. For each system, we attempt to solve it using Renegar's algorithm.
Since there exists at least one PBD $\q$ close to $\p$ with
such a set of multiplicities, we are guaranteed to find a solution,
which in turn must describe a PBD close to $\p$.

One technical issue that arises in the above program occurs when $\var[\p] \ll \log(1/\epsilon)$.
In this case, the effective support of the Fourier transform of $\p$
cannot be restricted to a small subset. \new{This causes problems with the convergence of
our Taylor expansion of the log Fourier transform for parameters near $1/2$.}
However, then only $O(\log(1/\epsilon))$ parameters are not close to $0$ and $1$,
and we can deal with such parameters separately.

\subsection{Related Work.} \label{sec:related}
Distribution learning is a classical problem
in statistics with a rich history and extensive literature (see e.g.,~\cite{BBBB:72, DG85, Silverman:86,Scott:92,DL:01}).
During the past couple of decades, a body of work in theoretical computer science has been studying these questions from
a computational complexity perspective; see e.g.,
\cite{KMR+:94short,FreundMansour:99short,AroraKannan:01, CGG:02, VempalaWang:02,FOS:05focsshort, BelkinSinha:10, KMV:10,
MoitraValiant:10, DDS12stoc, DDOST13focs, CDSS14, CDSS14b, ADLS15}.

We remark that the majority of the literature has focused either on non-proper learning (density estimation) or on
parameter estimation.
Regarding proper learning, a number of recent works in the statistics community
have given proper learners for structured distribution
families, by using a maximum likelihood approach.
See e.g.,~\cite{DumbgenRufibach:09, GW09sc, Walther09, DossW13, ChenSam13, KimSam14, BalDoss14}
for the case of continuous log-concave densities.
Alas, the computational complexity of these approaches has not been
analyzed. Two recent works~\cite{ADK15, CDGR15} yield computationally efficient proper learners
for discrete log-concave distributions, by using an appropriate convex formulation.
Proper learning has also been recently studied in the context of mixture models~\cite{FOS:05focsshort, DK14, SOAJ14, LiS15a}.
Here, the underlying optimization problems are non-convex, and
efficient algorithms are known only when the number of mixture components is small.

\subsection{Organization.} In Section~\ref{sec:count}, we prove our main structural result,
and in Section~\ref{sec:alg}, we describe our algorithm and prove its correctness. In Section~\ref{sec:concl},
we conclude with some directions for future research.

\section{Main Structural Result} \label{sec:count}
In this section, we prove our main structural results thereby establishing Theorems~\ref{thm:simple-distinct-param}
and~\ref{thm:complex-distinct-param-informal}.
Our proofs rely on an analysis of the Fourier transform
of PBDs combined with recent results from algebraic geometry on the solution structure
of systems of symmetric polynomial equations.
We show the following:

\begin{theorem} \label{thm:complex-distinct-param}
Given any $n$-PBD $\p$ with $\var[\p] = \poly(1/\eps)$,
there is an $n$-PBD $\q$ with $\dtv(\p,\q) \leq \eps$ such that
$\E[\q]=\E[\p]$ and $\var[\p] - \eps^3 \leq \var[\q] \leq \var[\p]$,
satisfying the following properties:

Let $R \eqdef \min \{1/4, \sqrt{\ln (1/\eps)/\var[\p]} \}$.
Let $B_i \eqdef R^{2^i}$,  for the integers $0 \leq i \leq \ell$,
where $\ell=O(\log \log(1/\eps))$ is selected such that $B_{\ell} = \poly(\eps)$.
Consider the partition $\mathcal{I} = \{I_i, J_i \}_{i=0}^{\ell+1}$ of $[0, 1]$ into the following set of intervals:
$I_0 = [B_0, 1/2]$, $I_{i+1} = [B_{i+1},B_i)$, $0 \le i \le {\ell-1}$, $I_{\ell+1} = (0, B_{\ell})$;
and $J_0 = (1/2, 1-B_0]$,  $J_{i+1} = (1- B_{i}, 1-B_{i+1}]$,
$0 \le i \le {\ell-1}$, $J_{\ell+1} = (1-B_{\ell}, 1]$.
Then we have the following:
\begin{enumerate}
\item[(i)] For each $0 \leq i \leq \ell$, each of the intervals $I_{i}$ and $J_{i}$
contains at most $O(\log(1/\eps)/\log (1/B_i))$ distinct parameters of $\q$.

\item[(ii) ]$\q$ has at most one parameter in each of the intervals $I_{\ell+1}$ and $J_{\ell+1} \setminus \{1\}$.

\item[(iii)] The number of parameters of $\q$ equal to $1$ is within an additive $\poly(1/\eps)$ of $\E[\p]$.

\item[(iv)] For each $0 \leq i \leq \ell$, each of the intervals $I_{i}$ and $J_{i}$
contains at most $2\var[\p]/B_{i}$ parameters of $\q$.
\end{enumerate}
\end{theorem}

Theorem~\ref{thm:complex-distinct-param} implies
that one needs to only consider $(1/\eps)^{O(\log \log (1/\eps))}$ different combinations of
multiplicities:

\begin{lemma}\label{NumberofMultiplicitiesLem}
For every $\p$ as in Theorem \ref{thm:complex-distinct-param},
there exists an explicit set $\mathcal{M}$ of multisets of triples $(m_i,a_i,b_i)_{1\leq i \leq k}$ so that
\begin{enumerate}
\item[(i)] For each element of $\mathcal{M}$ and each $i$, $[a_i,b_i]$ is either one of the intervals $I_i$ or $J_i$ as in Theorem \ref{thm:complex-distinct-param} or $[0,0]$ or $[1,1]$.
\item[(ii)] For each element of $\mathcal{M}$, $k=O(\log(1/\epsilon))$.
\item[(iii)] There exist an element of $\mathcal{M}$ and a PBD $\q$ as in the statement of Theorem \ref{thm:complex-distinct-param} with $\dtv(\p,\q)<\epsilon^2$ so that $\q$ has a parameter of multiplicity $m_i$ between $a_i$ and $b_i$ for each $1\leq i \leq k$ and no other parameters.
\item [(iv)] $\mathcal{M}$ has size $\left(\frac{1}{\epsilon}\right)^{O(\log\log(1/\epsilon))}$ and can be enumerated in $\poly(|\mathcal{M}|)$ time.
\end{enumerate}
\end{lemma}

This is proved in Appendix \ref{ap:NumberofMultiplicitiesLem} by a simple counting argument. We multiply the number of multiplicities for each interval, which is at most the maximum number of parameters to the power of the maximum number of distinct parameters in that interval, giving $(1/\eps)^{O(\log\log(1/\epsilon))}$ possibilities.

\noindent We now proceed to prove Theorem~\ref{thm:complex-distinct-param}.
We will require the following result from algebraic geometry:
\begin{theorem}[Part of Theorem 4.2 from~\cite{R:11}]  \label{thm:reiner}
Given $m+1$ symmetric polynomials in $n$ variables
$F_j(x)$, $0 \le j \le m$, $x \in \R^n$,
let $K = \{ x \in \R^n \mid F_j(x) \geq 0, \textrm{for all } 1 \le j \le m \}$.
Let $k = \max \{2, \lceil \mathrm{deg} (F_0) /2 \rceil, \mathrm{deg} (F_1) , \mathrm{deg} (F_2) , \ldots, \mathrm{deg} (F_m) \}$.
Then, the minimum value of $F_0$ on $K$ is achieved by a point with at most $k$ distinct co-ordinates.
\end{theorem}

As an immediate corollary, we obtain the following:

\begin{corollary} \label{cor:distinct-equations-in-ab}
If a set of multivariate polynomial equations $F_i(x)=0$, $x \in \mathbb{R}^n$, $1 \leq i \leq m$,
with the degree of each $F_i(x)$ being at most $d$ has a solution $x \in [a,b]^n$,
then it has a solution $y \in [a,b]^n$ with at most $d$ distinct values of the variables in $y$.
\end{corollary}

The following lemma will be crucial:

\begin{lemma}\label{momentMatchLem}
Let $\epsilon>0$.
Let $\p$ and $\q$ be $n$-PBDs with $\p$ having parameters $p_1,\ldots,p_k\leq 1/2$ and $p_1',\ldots,p_m'>1/2$ and $\q$ having parameters $q_1,\ldots,q_k\leq 1/2$ and $q_1',\ldots,q_m'>1/2$. Suppose furthermore that $\var[\p]=\var[\q]=V$ and let $C>0$ be a sufficiently large constant.
Suppose furthermore that for {$A=\min(3,C\sqrt{\log(1/\epsilon)/V})$} and for all positive integers $\ell$ it holds
\begin{equation}\label{reqMomentBoundEqn}
A^\ell \left(\left|\sum_{i=1}^k p_i^\ell-\sum_{i=1}^k q_i^\ell \right|+\left|\sum_{i=1}^m (1-p'_i)^\ell-\sum_{i=1}^m (1-q'_i)^\ell \right|\right) <\epsilon/C\log(1/\epsilon).
\end{equation}
Then $\dtv(\p,\q)<\epsilon$.
\end{lemma}

In practice, we shall only need to deal with a finite number of $\ell$'s,
since we will be considering the case where all $p_i,q_i$ or $1-p'_i,1-q'_i$ that do not appear in pairs
will have size less than $1/(2A)$.
Therefore, the size of the sum in question will be sufficiently small automatically for $\ell$ larger than $\Omega(\log((k+m)/\eps))$.

The basic idea of the proof will be to show that the Fourier transforms of $\p$ and $\q$ are close to each other.
In particular, we will need to make use of the following intermediate lemma:

\begin{lemma} \label{closePBDLem}
Let $\p$, $\q$ be PBDs with $|\E[\p]-\E[\q]|=O(\var[\p]^{1/2})$ and
$\var[\p]+1=\Theta(\var[\q]+1)$.
Let $M=\Theta(\log(1/\epsilon)+\sqrt{\var[\p] \log(1/\epsilon)})$  and $\ell = \Theta(\log(1/\eps))$ be positive integers with the implied constants sufficiently large.
If $\sum_{-\ell \le \xi \le \ell} |\widehat{\p}(\xi)-\widehat{\q}(\xi)|^2 \le \eps^2/16$,
then $\dtv(\p,\q) \leq \eps.$
\end{lemma}
The proof of this lemma, which is given in Appendix \ref{ap:closePBDLem}, is similar
to (part of) the correctness analysis of the non-proper learning algorithm in~\cite{DKS15}.

\begin{proof}[{\bf Proof of Lemma \ref{momentMatchLem}}]
We proceed by means of Lemma \ref{closePBDLem}. We need only show that for all $\xi$ with $|\xi|=O(\log(1/\epsilon))$ that $|\widehat{\p}(\xi)-\widehat{\q}(\xi)|\ll \epsilon/\sqrt{\log(1/\epsilon)}.$ For this we note that
\begin{align*}
\widehat{\p}(\xi) & = \prod_{i=1}^k ((1-p_i) + p_i e(\xi/M)) \prod_{i=1}^m ((1-p'_i) + p'_i e(\xi/M)) \\
& = e(m\xi/M) \prod_{i=1}^k (1 + p_i (e(\xi/M)-1)) \prod_{i=1}^m (1+ (1-p'_i)(e(-\xi/M)-1)).
\end{align*}
Taking a logarithm and Taylor expanding, we find that
\begin{equation}\label{logTaylorEqn}
\log(\widehat{\p}(\xi)) = 2\pi i m\xi/M + \sum_{\ell=1}^\infty \frac{(-1)^{1+\ell}}{\ell} \left((e(\xi/M)-1)^\ell\sum_{i=1}^k p_i^\ell + (e(-\xi/M)-1)^\ell\sum_{i=1}^m (1-p'_i)^\ell \right).
\end{equation}
A similar formula holds for $\log(\widehat{\q}(\xi))$. Therefore, we have that
$$|\widehat{\p}(\xi)-\widehat{\q}(\xi)|\leq |\log(\widehat{\p}(\xi))-\log(\widehat{\q}(\xi))| \;,$$ which is at most
\begin{align*}
& \sum_{\ell=1}^\infty |e(\xi/M)-1|^\ell \left(\left|\sum_{i=1}^k p_i^\ell-\sum_{i=1}^k q_i^\ell \right|+\left|\sum_{i=1}^m (1-p'_i)^\ell-\sum_{i=1}^m (1-q'_i)^\ell \right|\right)\\
\leq & \sum_{\ell=1}^\infty (2A/3)^\ell \left(\left|\sum_{i=1}^k p_i^\ell-\sum_{i=1}^k q_i^\ell \right|+\left|\sum_{i=1}^m (1-p'_i)^\ell-\sum_{i=1}^m (1-q'_i)^\ell \right|\right)\\
\leq & \sum_{\ell=1}^\infty (2/3)^\ell \epsilon/C\log(1/\epsilon)\\
\ll & \epsilon/C\log(1/\epsilon).
\end{align*}
An application of Lemma \ref{closePBDLem} completes the proof.
\end{proof}


\begin{proof}[{\bf Proof of Theorem \ref{thm:complex-distinct-param}}]
The basic idea of the proof is as follows. First, we will show that it is possible to modify $\p$
in order to satisfy (ii) without changing its mean, increasing its variance (or decreasing it by too much),
or changing it substantially in total variation distance. Next, for each of the other intervals $I_i$ or $J_i$,
we will show that it is possible to modify the parameters that $\p$ has in this interval
to have the appropriate number of distinct parameters,
without substantially changing the distribution in variation distance.
Once this holds for each $i$, conditions (iii) and (iv) will follow automatically.

To begin with, we modify $\p$ to have at most one parameter in $I_{\ell+1}$ in the following way.
We repeat the following procedure.
So long as $\p$ has two parameters, $p$ and $p'$ in $I_{\ell+1}$,
we replace those parameters by $0$ and $p+p'$. We note that this operation has the following properties:
\begin{itemize}
\item The expectation of $\p$ remains unchanged.
\item The total variation distance between the old and new distributions is $O(pp')$, as is the change in variances between the distributions.
\item The variance of $\p$ is decreased.
\item The number of parameters in $I_{\ell+1}$ is decreased by 1.
\end{itemize}
All of these properties are straightforward to verify by considering the effect of just the sum of the two changed variables.
By repeating this procedure, we eventually obtain a new PBD, $\p'$ with the same mean as $\p$, smaller variance,
and at most one parameter in $I_{\ell+1}$. We also claim that $\dtv(\p,\p')$ is small.
To show this, we note that in each replacement, the error in variation distance is at most a constant times
the increase in the sum of the squares of the parameters of the relevant PBD.
Therefore, letting $p_i$ be the parameters of $\p$ and letting $p'_i$ be the parameters of $\p'$,
we have that $\dtv(\p,\p')=O(\sum (p'_i)^2-p_i^2)$.
We note that this difference is entirely due to the parameters that were modified by this procedure.
Therefore, it is at most $(2B_\ell)^2$ times the number of non-zero parameters created.
Note that all but one of these parameters contributes at least $B_\ell/2$ to the variance of $\p'$.
Therefore, this number is at most $2\var[\p]/B_\ell +1$. Hence, the total variation distance between
$\p$ and $\p'$ is at most $O(B_\ell^2)(\var[\p]/B_\ell+1) \leq \epsilon^3.$ Similarly, the variance of our distribution is decreased by at most this much.
This implies that it suffices to consider $\p$ that have at most one parameter in $I_{\ell+1}$.
Symmetrically, we can also remove all but one of the parameters in $J_{\ell+1}$,
and thus it suffices to consider $\p$ that satisfy condition (ii).

Next, we show that for any such $\p$ that it is possible to modify the parameters that $\p$ has in $I_i$ or $J_i$, for any $i$,
so that we leave the expectation and variance unchanged, introduce at most $\epsilon^2$ error in variation distance,
and leave only $O(\log(1/\epsilon)/\log(1/B_{i}))$ distinct parameters in this range. The basic idea of this is as follows.
By Lemma \ref{momentMatchLem}, it suffices to keep $\sum p_i^\ell$ or $\sum (1-p_i)^\ell$
constant for parameters $p_i$ in that range for some range of values of $\ell$.
On the other hand, Theorem \ref{thm:reiner} implies that this can be done while producing only a small number of distinct parameters.

Without loss of generality assume that we are dealing with the interval $I_i$.
Note that if $i=0$ and $\var[\p]\ll \log(1/\epsilon)$, then $B_0=1/4$,
and there can be at most $O(\log(1/\epsilon))$ parameters in $I_0$ to begin with.
Hence, in this case there is nothing to show. Thus, assume that either $i\geq 0$ or that $\var[\p]\gg \log(1/\epsilon)$
with a sufficiently large constant. Let $p_1,\ldots,p_m$ be the parameters of $p$ that lie in $I_i$. Consider replacing them with parameters $q_1,\ldots, q_m$ also in $I_i$ to obtain $\q$. By Lemma \ref{momentMatchLem}, we have that $\dtv(\p,\q)<\epsilon^2$
so long as the first two moments of $\p$ and $\q$ agree and
\begin{equation}\label{replacementMomentEqn}
\min(3,C\sqrt{\log(1/\epsilon)/\var[\p]})^\ell\left|\sum_{j=1}^m p_j^\ell - \sum_{j=1}^m q_j^\ell\right| < \epsilon^3 \;,
\end{equation}
for all $\ell$ (the terms in the sum in Equation (\ref{reqMomentBoundEqn}) coming from the parameters not being changed cancel out). Note that $\min(3,C\sqrt{\log(1/\epsilon)/\var[\p]})\max(p_j,q_j)\leq B_i^{O(1)}$. This is because by assumption either $i>0$ and $\max(p_j,q_j)\leq \sqrt{B_i} \leq 1/4$ or $i=0$ and $B_i = \sqrt{\log(1/\epsilon)/\var[\p]} \ll 1$. Furthermore, note that $\var[\p]\geq m B_{i+1}$. Therefore, $m\leq \poly(1/\epsilon)$. Combining the above, we find that Equation (\ref{replacementMomentEqn}) is automatically satisfied for any $q_j\in I_i$ so long as $\ell$ is larger than a sufficiently large multiple of $\log(1/\epsilon)/\log(1/B_i)$.
On the other hand, Theorem \ref{thm:reiner} implies that there is some choice of $q_j\in I_i$ taking on only $O(\log(1/\epsilon)/\log(1/B_i))$ distinct values,
so that $\sum_{j=1}^m q_j^\ell$ is exactly $\sum_{j=1}^m p_j^\ell$ for all $\ell$ in this range.
Thus, replacing the $p_j$'s in this range by these $q_j$'s, we only change the total variation distance by $\epsilon^2$,
leave the expectation and variance the same (as we have fixed the first two moments),
and have changed our distribution in variation distance by at most $\epsilon^2$.

Repeating the above procedure for each interval $I_i$ or $J_i$ in turn, we replace $\p$ by a new PBD,
$\q$ with the same expectation and smaller variance and $\dtv(\p,\q)<\epsilon$,
so that $\q$ satisfies conditions (i) and (ii). We claim that (iii) and (iv) are necessarily satisfied.
Condition (iii) follows from noting that the number of parameters not 0 or 1 is at most $2+2\var[\p]/B_\ell$,
which is $\poly(1/\epsilon)$. Therefore, the expectation of $\q$ is the number of parameters
equal to $1+\poly(1/\epsilon)$. Condition (iv) follows upon noting that $\var[\q]\leq \var[\p]$
is at least the number of parameters in $I_i$ or $J_i$ times $B_i/2$ (as each contributes at least $B_i/2$ to the variance).
This completes the proof of Theorem~\ref{thm:complex-distinct-param}.
\end{proof}

\section{Proper Learning Algorithm} \label{sec:alg}

Given samples from an unknown PBD $\p$,
and given a collection of intervals and multiplicities
as described in Theorem \ref{thm:complex-distinct-param},
we wish to find a PBD $\q$ with those multiplicities that approximates $\p$.
By Lemma \ref{momentMatchLem}, it is sufficient to find such a $\q$
so that $\widehat{\q}(\xi)$ is close to $\widehat{\p}(\xi)$ for all small $\xi$.
On the other hand, by Equation (\ref{logTaylorEqn}) the logarithm of the Taylor series
of \new{$\widehat{\q}$}
is given by an appropriate expansion in the parameters.
Note that if $|\xi|$ is small, due to the $(e(\xi/M)-1)^\ell$ term,
the terms of our sum with $\ell \gg \log(1/\eps)$ will automatically be small.
By truncating the Taylor series, we get a polynomial in the parameters
that gives us an approximation to $\log(\widehat{\q}(\xi))$.
By applying a truncated Taylor series for the exponential function,
we obtain a polynomial in the parameters of $\q$ which approximates its Fourier coefficients.
This procedure yields a system of polynomial equations whose solution
gives the parameters of a PBD that approximates $\p$.
Our main technique will be to solve this system of equations to obtain our output distribution using the following result:

\begin{theorem}[\cite{Renegar92-sicomp, Renegar92-jsc}] \label{renegar}
Let $P_i: \R^n \to \R$, $i=1, \ldots, m$, be $m$ polynomials over the reals each of maximum degree at most $d$.
Let $K = \{x \in \R^n: P_i(x) \ge 0, \textrm{ for all } i=1, \ldots,  m\}$. If the coefficients of the $P_i$'s are rational numbers with bit complexity at most $L$,
there is an algorithm that runs in time $\poly(L, (d \cdot m)^n)$ and decides if $K$ is empty or not. Further, if $K$ is non-empty, the algorithm runs in time
$\poly(L, (d \cdot m)^n, \log(1/\delta))$ and outputs a point in $K$ up to an $L_2$ error $\delta$.
\end{theorem}

In order to set up the necessary system of polynomial equations, we have the following theorem:

\begin{theorem}\label{systemThm}
Consider a PBD $\p$ with $\var[\p]<\poly(1/\epsilon)$,
and real numbers $\tilde \sigma \in [\sqrt{\var[\p]}/2,2\sqrt{\var[\p]}+1]$
and $\tilde \mu$ with $|\E[\p]-\tilde \mu|\leq \tilde\sigma$.
Let $M$ be as above and let $\ell$ be a sufficiently large multiple of $\log(1/\epsilon)$.
Let $h_\xi$ be complex numbers for each integer $\xi$
with $|\xi| \le \ell$ so that $\sum_{|\xi|\leq \ell} |h_\xi-\widehat{\p}(\xi)|^2 < \epsilon^2/16.$

Consider another PBD with parameters $q_i$ of multiplicity $m_i$ contained in intervals $[a_i,b_i]$ as described in Theorem \ref{thm:complex-distinct-param}. There exists an explicit system $\mathcal{P}$ of $O(\log(1/\epsilon))$ real polynomial inequalities each of degree $O(\log(1/\epsilon))$ in the $q_i$ so that:
\begin{itemize}
\item[(i)] If there exists such a PBD of the form of $\q$ with $\dtv(\p,\q) < \epsilon/\ell$, $\E[\q]=\E[\p]$, and $\var[\p]\geq \var[\q] \geq \var[\p]/2$, then its parameters $q_i$ yield a solution to $\mathcal{P}$.
\item[(ii)] Any solution $\{q_i\}$ to $\mathcal{P}$ corresponds to a PBD $\q$ with $\dtv(\p,\q) < \epsilon/2.$
\end{itemize}
Furthermore, such a system can be found with rational coefficients of encoding size $O(\log^2(1/\eps))$ bits.
\end{theorem}

\begin{proof}
For technical reasons, we begin by considering the case that $\var[\p]$
is larger than a sufficiently large multiple of $\log(1/\epsilon)$,
as we will need to make use of slightly different techniques in the other case.
In this case, we construct our system $\cal P$ in the following manner.
We begin by putting appropriate constraints on the mean and variance of $\q$ and requiring that the $q_i$'s lie in appropriate intervals.
\begin{eqnarray}
\widetilde{\mu} - 2\widetilde{\sigma} \leq \sum_{j=1}^k m_j p_j \leq \widetilde{\mu} + 2\widetilde{\sigma} \label{eqn:mean-approx} \\
\widetilde{\sigma}^2/2 -1 \leq \sum_{j=1}^k m_j p_j (1-p_j) \leq 2 \widetilde{\sigma}^2  \label{eqn:variance-approx} \\
a_j \leq p_j \leq b_j, \label{eqn:interval}
\end{eqnarray}
Next, we need a low-degree polynomial to express the condition that Fourier coefficients of $\q$ are approximately correct.
To do this, we let $S$ denote the set of indices $i$ so that $[a_i,b_i]\subset [0,1/2]$ and $T$ the set so that $[a_i,b_i]\subset [1/2,1]$ and let $m=\sum_{i\in T}m_i$.
We let
\begin{equation}\label{gdefEqn}
g_\xi = 2\pi i \xi m/M+\sum_{k=1}^\ell \frac{(-1)^{k+1}}{k}\left((e(\xi/M)-1)^k\sum_{i\in S}m_i q_i^k+(e(-\xi/M)-1)^k \sum_{i\in T} m_i(1-q_i)^k \right)
\end{equation}
be an approximation to the logarithm of $\widehat{\q}(\xi)$. We next define $\exp'$ to be a Taylor approximation to the exponential function
$$
\exp'(z) := \sum_{k=0}^\ell \frac{z^k}{k!}.
$$
By Taylor's theorem, we have that
$$
|\exp'(z)-\exp(z)| \leq \frac{z^{\ell+1}\exp(z)}{(\ell+1)!},
$$
and in particular that if $|z| < \ell/3$ that $|\exp'(z)-\exp(z)| = \exp(-\Omega(\ell))$.

We would ideally like to use $\exp'(g_\xi)$ as an approximation to $\widehat{\q}(\xi)$.
Unfortunately, $g_\xi$ may have a large imaginary part.
To overcome this issue, we let $o_\xi$, defined as the nearest integer to $\tilde \mu \xi/M$,
be an approximation to the imaginary part, and we set
\begin{equation}\label{qdefEqn}
q_\xi = \exp'(g_\xi+2\pi i o_\xi) \;.
\end{equation}
We complete our system $\mathcal{P}$ with the final inequality:
\begin{equation}
\sum_{-\ell \leq \xi \leq \ell} |q_\xi - h_\xi|^2 \leq  \eps^2/8 \label{eqn:ft-approx}.
\end{equation}

In order for our analysis to work, we will need for $q_\xi$ to approximate $\widehat{\q}(\xi)$. Thus, we make the following claim:
\begin{claim}\label{qApproxClaim}
If Equations (\ref{eqn:mean-approx}), (\ref{eqn:variance-approx}), (\ref{eqn:interval}), (\ref{gdefEqn}), and (\ref{qdefEqn}) hold,
then $|q_\xi - \widehat{\q}(\xi)|<\epsilon^3$ for all $|\xi|\leq \ell$.
\end{claim}
This is proved in Appendix \ref{apsec:alg} by showing that $g_\xi$
is close to a branch of the logarithm of $\widehat{\q}(\xi)$ and that $|g_\xi+2\pi i o_\xi| \leq O(\log(1/\eps))$, so $\exp'$
 is a good enough approximation to the exponential.

Hence, our system $\mathcal{P}$ is defined as follows:

\noindent\textbf{Variables:}
\begin{itemize}
\item $q_i$ for each distinct parameter $i$ of $\q$.
\item $g_\xi$ for each $|\xi|\leq \ell$.
\item $q_\xi$ for each $|\xi|\leq \ell$.
\end{itemize}
\noindent\textbf{Equations:} Equations (\ref{eqn:mean-approx}), (\ref{eqn:variance-approx}), (\ref{eqn:interval}), (\ref{gdefEqn}), (\ref{qdefEqn}), and (\ref{eqn:ft-approx}).

\smallskip

To prove (i), we note that such a $\q$ will satisfy (\ref{eqn:mean-approx}) and (\ref{eqn:variance-approx}),
because of the bounds on its mean and variance,
and will satisfy Equation (\ref{eqn:interval}) by assumption.
Therefore, by Claim \ref{qApproxClaim}, $q_\xi$ is approximately $\widehat{\q}(\xi)$ for all $\xi$. On the other hand, since $\dtv(\p,\q)<\epsilon/\ell$, we have that $|\widehat{\p}(\xi)-\widehat{\q}(\xi)|<\epsilon/\ell$ for all $\xi$. Therefore, setting $g_\xi$ and $q_\xi$ as specified, Equation (\ref{eqn:ft-approx}) follows.
To prove (ii), we note that a $\q$ whose parameters satisfy $\p$ will by Claim \ref{qApproxClaim} satisfy the hypotheses of {Lemma \ref{closePBDLem}}.
Therefore, $\dtv(\p,\q)\leq \epsilon/2.$

{As we have defined it so far,
the system $\mathcal{P}$ does not have rational coefficients.
Equation (\ref{gdefEqn}) makes use of $e(\pm \xi/M)$ and $\pi$,
as does Equation (\ref{qdefEqn}). To fix this issue, we note that
if we approximate the appropriate powers of $(\pm 1 \pm e(\pm \xi/M))$ and $q\pi i$
each to accuracy $(\epsilon/\sum_{i \in S} m_i))^{10}$,
this produces an error of size at most $\epsilon^4$ in the value $g_\xi$, and therefore
an error of size at most $\epsilon^3$ for $q_\xi$, and this leaves the above argument unchanged.

Also, as defined above,
the system $\mathcal{P}$ has complex constants
and variables and many of the equations equate complex quantities.
The system can be expressed as a set of real inequalities by doubling
the number of equations and variables
to deal with the real and imaginary parts separately.
Doing so introduces binomial coefficients into the coefficients, which
are no bigger than $2^{O(\log(1/\eps))}=\poly(1/\eps)$ in magnitude.
To express $\exp'$, we need denominators with a factor of $\ell!=\log(1/\eps)^{\Theta(\log(1/\eps))}$.
All other constants can be expressed as rationals with numerator and denominator bounded by $\poly(1/\eps)$.
So, the encoding size of any of the rationals that appear in the system is $\log(\log(1/\eps)^{O(\log(1/\eps))})=O(\log^2(1/\eps))$.

}

One slightly more difficult problem is that the proof of Claim \ref{qApproxClaim} depended upon the fact that $\var[\p] \gg \log(1/\epsilon)$.
If this is not the case, we will in fact need to slightly modify our system of equations. In particular, we redefine $S$ to be the set of indices, $i$, so that $b_i\leq 1/4$ (rather than $\leq 1/2$), and let $T$ be the set of indices $i$ so that $a_i \geq 3/4$. Finally, we let $R$ be the set of indices for which $[a_i,b_i]\subset [1/4,3/4]$. We note that,
since each $i\in R$ contributes at least $m_i/8$ to $\sum_i m_i q_i(1-q_i)$, if Equations (\ref{eqn:interval}) and (\ref{eqn:variance-approx}) both hold,
we must have $|R|=O(\var[\p]) = O(\log(1/\epsilon))$.

We then slightly modify Equation (\ref{qdefEqn}), replacing it by
\begin{equation}\label{NewqdefEqn}
q_\xi = \exp'(g_\xi)\prod_{i\in R}(q_ie(\xi/M) + (1-q_i))^{m_i}.
\end{equation}
Note that by our bound on $\sum_{i\in R} m_i$, this is of degree $O(\log(1/\epsilon))$.

We now need only prove the analogue of Claim \ref{qApproxClaim} in order for the rest of our analysis to follow.
\begin{claim} \label{qApproxClaim-small}
If Equations (\ref{eqn:mean-approx}), (\ref{eqn:variance-approx}), (\ref{eqn:interval}), (\ref{gdefEqn}),  and (\ref{NewqdefEqn}) hold,
then $|q_\xi - \widehat{\q}(\xi)|<\epsilon^3$ for all $|\xi|\leq \ell$.
\end{claim}
We prove this in Appendix \ref{apsec:alg}, by proving similar bounds to those needed for Claim \ref{qApproxClaim}.
This completes the proof of our theorem in the second case.
\end{proof}

Our algorithm for properly learning PBDs is given in pseudocode below:

\vspace{0.2cm}

\fbox{\parbox{6.3in}{
{\bf Algorithm} {\tt Proper-Learn-PBD}\\
Input: sample access to a PBD $\p$ and $\eps>0$.\\
Output: A hypothesis PBD that is $\eps$-close to $\p$ with probability at least $9/10$. \\

\vspace{-0.2cm}

 Let $C$ be a sufficiently large universal constant.

\vspace{-0.2cm}

\begin{enumerate}

\item
Draw $O(1)$ samples from $\p$
and with confidence probability $19/20$ compute: (a) $\widetilde{\sigma}^2$, a factor $2$ approximation to $\var_{X \sim \p}[X]+1$,
and (b) $\widetilde{\mu}$, an approximation to $\E_{X \sim \p}[X]$ to within one standard deviation.
Set $M \eqdef {\lceil C(\log(1/\epsilon)+ \widetilde{\sigma}\sqrt{\log(1/\eps)}) \rceil}$.
Let $\ell \eqdef \lceil C^2 \log(1/\eps) \rceil$. 

\item If $\widetilde{\sigma} > \Omega(1/\eps^3),$
then we draw $O(1/\eps^2)$ samples
and use them to learn a shifted binomial distribution,
using algorithms {\tt Learn-Poisson} and {\tt Locate-Binomial} from \cite{DDS15-journal}.
Otherwise, we proceed as follows:

\item Draw $N=C^3 (1/\eps^2) \ln^{2}(1/\eps)$ samples $s_1,\ldots, s_N$ from $\p$.
For integers $\xi$ with $|\xi| \leq \ell$, set $h_\xi$ to be the empirical DFT modulo $M$.
Namely, $h_\xi := \frac{1}{N}\sum_{i=1}^N e(-\xi s_i/M).$

\item Let $\mathcal{M}$ be the set of multisets of multiplicities described in Lemma \ref{NumberofMultiplicitiesLem}.
For each element $m\in\mathcal{M},$
let $\mathcal{P}_m$ be the corresponding
system of polynomial equations as described in Theorem \ref{systemThm}.

\item For each such system, use the algorithm from {Theorem \ref{renegar}}
to find a solution to precision $\epsilon/\new{(2k)}$,
{where $k$ is the sum of the multiplicities not corresponding to $0$ or $1$},
if such a solution exists. Once such a solution is found,
return the PBD $\q$ with parameters $q_i$ to multiplicity $m_i$,
where $m_i$ are the terms from $m$ and $q_i$ in the approximate solution to $\mathcal{P}_m$.
\end{enumerate}
}}

\vspace{0.2cm}

\begin{proof}[{\bf Proof of Theorem~\ref{thm:proper-learning}}]

{
We first note that the algorithm succeeds in the case that $\var_{X \sim \p}[X] = \Omega(1/\eps^6)$:
\cite{DDS15-journal} describes procedures {\tt Learn-Poisson} and {\tt Locate-Binomial} that draw O($1/\eps^2$) samples,
and return a shifted binomial $\eps$-close to a PBD $\p$, provided $\p$ is not close to a PBD in ``sparse form'' in their terminology.
This hods for any PBD with effective support $\Omega(1/\eps^3)$, since by definition a PBD in ``sparse form'' has
support of size $O(1/\eps^3)$.}

It is clear that the sample complexity of our algorithm is $O(\eps^{-2} \log^2(1/\epsilon))$.
The runtime of the algorithm is dominated by Step~5.
We note that by Lemma \ref{NumberofMultiplicitiesLem}, $|\mathcal{M}| = (1/\epsilon)^{O(\log\log(1/\epsilon))}$. 
Furthermore, by {Theorems \ref{renegar}} and \ref{systemThm}, the runtime for solving the system $\mathcal{P}_m$ 
is $O(\log(1/\epsilon))^{O(\log(1/\epsilon))}=(1/\epsilon)^{O(\log\log(1/\epsilon))}$. Therefore, the total
runtime is $(1/\epsilon)^{O(\log\log(1/\epsilon))}$.

It remains to show correctness. We first note that each $h_\xi$ is an average of independent random variables $e(-\xi p_i/M)$, 
with expectation $\widehat{\p}(\xi)$. Therefore, by standard Chernoff bounds, with high probability 
we have that $|h_\xi-\widehat{\p}(\xi)| = O(\sqrt{\log(\ell)}/\sqrt{N}) \ll \epsilon/\sqrt{\ell}$ for all $\xi$, and therefore we have that
$$
\sum_{|\xi|\leq \ell} |h_\xi - \widehat{\p}(\xi)|^2 < \epsilon^2/8.
$$
Now, by Lemma \ref{NumberofMultiplicitiesLem}, for some $m\in\mathcal{M}$ there will exist a PBD $\q$ 
whose distinct parameters come in multiplicities given by $m$ and lie in the corresponding intervals 
so that $\dtv(\p,\q)\leq \epsilon^2$. Therefore, by Theorem \ref{systemThm}, the system $\mathcal{P}_m$ will have a solution. 
Therefore, at least one $\mathcal{P}_m$ will have a solution and our algorithm will necessarily return \emph{some} PBD $\q.$

On the other hand, any $\q$ returned by our algorithm will correspond
to an approximation of some solution of $\mathcal{P}_m$,
for some $m\in\mathcal{M}$. By Theorem \ref{systemThm}, any solution to any $\mathcal{P}_m$
will give a PBD $\q$ with $\dtv(\p,\q)\leq\epsilon/2$.
Therefore, the actual output of our algorithm is a PBD $\q'$,
whose parameters approximate those of such a $\q$ to within $\epsilon/\new{(2k)}$.
On the other hand, from this it is clear that $\dtv(\q,\q')\leq \epsilon/2$,
and therefore, $\dtv(\p,\q')\leq \epsilon$. In conclusion,
our algorithm will always return a PBD that is within $\epsilon$ total variation distance of $\p$.
\end{proof}

\section{Conclusions and Open Problems} \label{sec:concl}

In this work, we gave a nearly-sample optimal algorithm for properly learning PBDs that runs in almost polynomial time.
We also provided a structural characterization for PBDs that may be of independent interest.
The obvious open problem is to obtain a polynomial-time proper learning algorithm.
We conjecture that such an algorithm is possible, and our mildly super-polynomial runtime
may be viewed as an indication of the plausibility of this conjecture.
Currently, we do not know of a  $\poly(1/\eps)$  time algorithm even for the special case of an $n$-PBD
with $n = O(\log(1/\eps)).$

A related open question concerns obtaining faster proper algorithms for learning
more general families of discrete distributions that are amenable to similar techniques,
e.g., sums of independent integer-valued random variables~\cite{DDOST13focs, DKS15}, and Poisson multinomial distributions~\cite{DKT15, DKS15b}.
Here, we believe that progress is attainable via a generalization of our techniques.

The recently obtained cover size lower bound for PBDs~\cite{DKS15} is a bottleneck for other non-convex optimization problems as well,
e.g., the problem of computing approximate Nash equilibria in anonymous games~\cite{DaskalakisP2014}. The fastest known algorithms
for these problems proceed by enumerating over an $\eps$-cover.
Can we obtain faster algorithms in such settings, by avoiding enumeration over a cover?

\newpage

\bibliographystyle{alpha}

\bibliography{allrefs}

\appendix

\section*{Appendix}

\section{Sample Complexity Lower Bound for Parameter Estimation} \label{sec:param}

\begin{proposition} \label{prop:param}
Suppose that $n\geq 1/\epsilon$.
Any learning algorithm that takes $N$ samples from an $n$-PBD
and returns estimates of these parameters to additive error at most $\epsilon$ with probability at least $2/3$
must have $N \geq 2^{\Omega(1/\epsilon)}$.
\end{proposition}
\begin{proof}
We may assume that $n=\Theta(1/\epsilon)$ (as we could always make the remaining parameters all $0$)
and demonstrate a pair of PBDs whose parameters differ by $\Omega(\epsilon)$,
and yet have variation distance $2^{-\Omega(1/\epsilon)}$. Therefore, if such an algorithm is given one of these two PBDs,
it will be unable to distinguish which one it is given, and therefore unable to learn the parameters to $\epsilon$ accuracy
with at least $2^{\Omega(1/\epsilon)}$ samples.

In order to make this construction work, we take $\p$ to have parameters $p_j:=(1+\cos\left(\frac{2\pi j}{n} \right))/8$,
and let $\q$ have parameters $q_j:=(1+\cos\left(\frac{2\pi j+\pi}{n} \right))/8$.
Suppose that $j=n/4+O(1)$. We claim that none of the $q_i$ are closer to $p_j$ that $\Omega(1/n)$.
This is because for all $i$ we have that $\left(\frac{2\pi i+\pi}{n}\right)$ is at least $\Omega(1/n)$ from $\left(\frac{2\pi j}{n}\right)$ and $\left(\frac{2\pi (n-j)}{n}\right)$.

On the other hand, it is easy to see that the $p_j$ are roots of the polynomial $(T_n(8x-1)-1)$, and $q_j$ are the roots of $(T_n(8x-1)+1)$, where $T_n$ is the $n^{th}$ Chebyshev polynomial. Since these polynomials have the same leading term and identical coefficients other than their constant terms, it follows
that the elementary symmetric polynomials in $p_j$ of degree less than $n$ equal the corresponding polynomials in the $q_j$. From this, by the Newton-Girard formulae, we have that $\sum_{i=1}^n p_i^l=\sum_{i=1}^n q_i^l$ for $1 \leq l \leq n-1$. For any $l \geq n$, we have that $3^l(\sum_{i=1}^n (p_i^l-q_i^l)) \leq n(3/4)^n$,
and so by Lemma~\ref{momentMatchLem}, we have that $\dtv(\p,\q)=2^{-\Omega(n)}$. This completes our proof.
\end{proof}

\section{Omitted Proofs from Section~\ref{sec:count}}

\subsection{Proof of Lemma~\ref{NumberofMultiplicitiesLem}.} \label{ap:NumberofMultiplicitiesLem}

For completeness, we restate the lemma below.

\medskip

\noindent {\bf Lemma \ref{NumberofMultiplicitiesLem}.}
{\em For every $\p$ as in Theorem \ref{thm:complex-distinct-param},
there exists an explicit set $\mathcal{M}$ of multisets of triples $(m_i,a_i,b_i)_{1\leq i \leq k}$ so that
\begin{enumerate}
\item[(i)] For each element of $\mathcal{M}$ and each $i$, $[a_i,b_i]$ is either one of the intervals $I_i$ or $J_i$ as in Theorem \ref{thm:complex-distinct-param} or $[0,0]$ or $[1,1]$.
\item[(ii)] For each element of $\mathcal{M}$, $k=O(\log(1/\epsilon))$.
\item[(iii)] There exist an element of $\mathcal{M}$ and a PBD $\q$ as in the statement of Theorem \ref{thm:complex-distinct-param} with $\dtv(\p,\q)<\epsilon^2$ so that $\q$ has a parameter of multiplicity $m_i$ between $a_i$ and $b_i$ for each $1\leq i \leq k$ and no other parameters.
\item [(iv)] $\mathcal{M}$ has size $\left(\frac{1}{\epsilon}\right)^{O(\log\log(1/\epsilon))}$ and can be enumerated in $\poly(|\mathcal{M}|)$ time.
\end{enumerate} }

\begin{proof}[Proof of Lemma~\ref{NumberofMultiplicitiesLem} assuming Theorem \ref{thm:complex-distinct-param}]
Replacing $\epsilon$ in Theorem \ref{thm:complex-distinct-param} by $\epsilon^2$, we take $\mathcal{M}$ to be the set of all possible ways to have at most $O(\log(1/\epsilon)/\log(1/B_i))$ terms with $[a_i,b_i]$ equal to $I_i$ or $J_i$ and having the sum of the corresponding $m$'s at most $4\var[\p]/B_i$, having one term with $a_i=b_i=1$ and $m_i=\E[\p]+\poly(1/\epsilon)$, and one term with $a_i=b_i=0$ and $m_i$ such that the sum of all of the $m_i$'s equals $n$.

For this choice of $\mathcal{M}$, (i) is automatically satisfied, and (iii) follows immediately from Theorem \ref{thm:complex-distinct-param}. To see (ii), we note that the total number of term in an element of $\mathcal{M}$ is at most
$$
O(1) + \sum_{i=1}^\ell O(\log(1/\epsilon)/\log(1/B_i)) = O(1) + \sum_{i=1}^\ell O(\log(1/\epsilon) 2^{-i}) = O(\log(1/\epsilon)).
$$

To see (iv), we need a slightly more complicated counting argument.
To enumerate $\mathcal{M}$, we merely need to enumerate each integer of size
$\E[\p]+\poly(1/\epsilon)$ for the number of $1$'s,
and enumerate for each $0\leq i \leq \ell$ all possible multi-sets of $m_i$ of size
at most $O(\log(1/\epsilon)/\log(1/B_i))$ with sum at most $2\var[\p]/B_i$
to correspond to the terms with $[a_i,b_i]=I_i$,
and again for the terms with $[a_i,b_i]=J_i$.
This is clearly enumerable in $\poly(|\mathcal{M}|)$ time,
and the total number of possible multi-sets is at most
$$
\poly(1/\epsilon)\prod_{i=0}^\ell (2\var[\p]/B_i)^{O(\log(1/\epsilon)/\log(1/B_i))}.
$$
Therefore, we have that
\begin{align*}
|\mathcal{M}| & \leq \poly(1/\epsilon)\prod_{i=0}^\ell (2\var[\p]/B_i)^{O(\log(1/\epsilon)/\log(1/B_i))}\\
& = \poly(1/\epsilon) \prod_{i=0}^\ell B_i^{-O(\log_{1/B_i}(1/\epsilon))} \prod_{i=0}^\ell O(\var[\p])^{O(\log(1/\epsilon)/(2^i \log(1/B_0)))}\\
& = \poly(1/\epsilon) \prod_{i=0}^\ell \poly(1/\epsilon) O(\var[\p])^{O(\log(1/\epsilon)/\log(1/B_0))}\\
& = (1/\epsilon)^{O(\log\log(1/\epsilon))}O(\var[\p])^{O(\log(1/\epsilon)/\log(1/B_0))}\\
& = (1/\epsilon)^{O(\log\log(1/\epsilon))}.
\end{align*}
The last equality above requires some explanation. If $\var[\p]<\log^2(1/\epsilon)$, then
$$
O(\var[\p])^{O(\log(1/\epsilon)/\log(1/B_0))} \leq \log(1/\epsilon)^{O(\log(1/\epsilon))} = (1/\epsilon)^{O(\log\log(1/\epsilon))}.
$$
Otherwise, if $\var[\p]\geq \log^2(1/\epsilon)$, $\log(1/B_0) \gg \log(\var[\p])$, and thus
$$
O(\var[\p])^{O(\log(1/\epsilon)/\log(1/B_0))} \leq \poly(1/\epsilon).
$$

This completes our proof.
\end{proof}

\subsection{Proof of Lemma~\ref{closePBDLem}.} \label{ap:closePBDLem}

For completeness, we restate the lemma below.

\medskip

\noindent {\bf Lemma \ref{closePBDLem}.}
{\em Let $\p$, $\q$ be PBDs with $|\E[\p]-\E[\q]|=O(\var[\p]^{1/2})$ and $\var[\p]=\Theta(\var[\q])$.
Let $M=\Theta({\log(1/\epsilon)}+\sqrt{\var[\p] \log(1/\epsilon)})$  and $\ell = \Theta(\log(1/\eps))$ be positive integers with the implied constants sufficiently large.
If $\sum_{-\ell \le \xi \le \ell} |\widehat{\p}(\xi)-\widehat{\q}(\xi)|^2 \le \eps^2/16$,
then $\dtv(\p,\q) \leq \eps.$ }

\begin{proof}
The proof of this lemma is similar to the analysis of correctness of the non-proper learning algorithm in~\cite{DKS15}.

The basic idea of the proof is as follows. {By Bernstein's inequality}, $\p$ and $\q$ both have nearly all of their probability mass supported in the same interval of length $M$. This means that is suffices to show that the distributions $\p \pmod{M}$ and $\q\pmod{M}$ are close. By Plancherel's Theorem, it suffices to show that
the DFTs $\widehat{\p}$ and $\widehat{\q}$ are close. However, it follows by Lemma 6 of \cite{DKS15} that these DFTs are small in magnitude
outside of $-\ell \le \xi \le \ell$.

Let $m$ be the nearest integer to the expected value of $\p$. {By Bernstein's inequality}, it follows that both $\p$ and $\q$ have $1-\epsilon/10$ of their probability mass in the interval $I=[m-M/2,m+M/2)$. We note that any given probability distribution $X$ over $\Z/M\Z$ has a unique lift to a distribution taking values in $I$. We claim that $\dtv(\p,\q)\leq \epsilon/5+\dtv(\p\pmod{M},\q\pmod{M})$. This is because after throwing away the at most $\epsilon/5$ probability mass where $\p$ or $\q$ take values outside of $I$, there is a one-to-one mapping between values in $I$ taken by $\p$ or $\q$ and the values taken by $\p\pmod{M}$ or $\q\pmod{M}$. Thus, it suffices to show that $\dtv(\p\pmod{M},\q\pmod{M})\leq 4\epsilon/5$.

By Cauchy-Schwarz, we have that
$$\dtv(\p\pmod{M},\q\pmod{M}) \leq \sqrt{M}\|\p\pmod{M}-\q\pmod{M}\|_2 \;.$$
By Plancherel's Theorem, the RHS above is
\begin{equation}\label{FTSumEqu}
\sqrt{\sum_{\xi\pmod{M}}|\widehat{\p}(\xi)-\widehat{\q}(\xi)|^2}.
\end{equation}
By assumption, the sum of the above over all $|\xi|\leq \ell$ is at most $\epsilon^2/16.$ 
However, applying Lemma 6 of \cite{DKS15} with $k=2$, we find that for any $|\xi|\leq M/2$ 
that each of $|\widehat{\p}(\xi)|,|\widehat{\q}(\xi)|$ is $\exp(-\Omega(\xi^2\var[\p]/M^2))=\exp(-\Omega(\xi^2/\log(1/\epsilon)))$. 
Therefore, the sum above over $\xi$ not within $\ell$ of some multiple of $M$ is at most
\begin{eqnarray*}
\sum_{n>\ell} \exp(-\Omega(n^2/\log(1/\epsilon))) &\leq& \sum_{n>\ell} \exp(-\Omega((\ell^2+(n-\ell)\ell)/\log(1/\epsilon))) \\
&\leq& \sum_{n>\ell} \exp(-(n-\ell))\exp(-\Omega(\ell^2/\log(1/\epsilon))) \leq \epsilon^2/16
\end{eqnarray*}
assuming that the constant defining $\ell$ is large enough. Therefore, the sum in (\ref{FTSumEqu}) is at most $\epsilon^2/8.$ This completes the proof.
\end{proof}

\section{Omitted Proofs from Section~\ref{sec:alg}} \label{apsec:alg}

In this section, we prove Claims \ref{qApproxClaim} and \ref{qApproxClaim-small} which we restate here.

\medskip

\noindent {\bf Claim \ref{qApproxClaim}.}
{\em If Equations (\ref{eqn:mean-approx}), (\ref{eqn:variance-approx}), (\ref{eqn:interval}), (\ref{gdefEqn}),  and (\ref{qdefEqn}) hold,
then $|q_\xi - \widehat{\q}(\xi)|<\epsilon^3$ for all $|\xi|\leq \ell$. }

\medskip

\begin{proof}
First we begin by showing that $g_\xi$ approximates $\log(\widehat{\q}(\xi))$. By Equation (\ref{logTaylorEqn}), 
we would have equality if the sum over $k$ were extended to all positive integers. 
Therefore, the error between $g_\xi$ and $\log(\widehat{\q}(\xi))$ is equal to the sum over all $k> \ell$. 
Since $\tilde \sigma \gg \log(1/\epsilon)$, we have that $M\gg \ell$ and therefore, $|1-e(\xi/m)|$ and $|e(-\xi/M)-1|$ are both less than $1/2.$ 
Therefore, the term for a particular value of $k$ is at most $2^{-k}\left(\sum_{i\in S} m_i q_i + \sum_{i\in T}m_i (1-q_i) \right) \gg 2^{-k}\tilde \sigma.$ 
Summing over $k> \ell$, we find that
$$
|g_\xi - \log(\widehat{\q}(\xi))| < \epsilon^4.
$$

We have left to prove that $\exp'(g_\xi - 2\pi i o_\xi)$ is approximately $\exp(g_\xi)=\exp(g_\xi-2\pi i o_\xi)$. By the above, it suffices to prove that $|g_\xi - 2\pi i o_\xi| < \ell/3.$ We note that
\begin{align*}
g_\xi & = 2\pi i \xi m/M+\sum_{k=1}^\ell \frac{(-1)^{k+1}}{k}\left((e(\xi/M)-1)^k\sum_{i\in S}m_i q_i^k+(e(-\xi/M)-1)^k \sum_{i\in T} m_i(1-q_i)^k \right)\\
& = 2\pi i \xi m/M+(e(\xi/M)-1)\sum_{i\in S}m_i q_i+(e(-\xi/M)-1)\sum_{i\in T} m_i(1-q_i) + \\
& + O\left(\sum_{k=2}^\ell |\xi|^2/M^2 2^{-k}\left(\sum_{i} m_i q_i(1-q_i) \right) \right)\\
& = 2\pi i \xi m/M+2\pi i \xi/M \left(\sum_{i\in S}m_i q_i -\sum_{i\in T} m_i(1-q_i) \right)+O(|\xi|^2/M^2 \tilde \sigma^2)\\
& = 2\pi i \xi/M \sum_i m_i q_i +O(|\xi|^2/M^2 \tilde \sigma^2)\\
& = 2\pi i \xi/M \tilde \mu +O(|\xi|/M \tilde \sigma)+O(|\xi|^2/M^2 \tilde \sigma^2)\\
& = 2\pi i o_\xi + O(\log(1/\eps)).
\end{align*}
This completes the proof.
\end{proof}

\medskip

\noindent {\bf Claim \ref{qApproxClaim-small}.}
{\em If Equations (\ref{eqn:mean-approx}), (\ref{eqn:variance-approx}), (\ref{eqn:interval}), (\ref{gdefEqn}), and (\ref{NewqdefEqn}) hold,
then $|q_\xi - \widehat{\q}(\xi)|<\epsilon^3$ for all $|\xi|\leq \ell$. }

\medskip

\begin{proof}
Let $\q'$ be the PBD obtained from $\q$ upon removing all parameters corresponding to elements of $R$. We note that
$$\widehat{\q}(\xi) = \widehat{\q'}(\xi)\prod_{i\in R}(q_ie(\xi/M) + (1-q_i))^{m_i}.$$
Therefore, it suffices to prove our claim when $R=\emptyset$.

Once again it suffices to show that $g_\xi$ is within $\epsilon^4$ of $\log(\widehat{\q}(\xi))$
and that $|g_\xi|<\ell/3$. For the former claim, we again note that, by Equation (\ref{logTaylorEqn}),
we would have equality if the sum over $k$ were extended to all integers,
and therefore only need to bound the sum over all $k>\ell$.
On the other hand, we note that $q_i\leq 1/4$ for $i\in S$ and $(1-q_i)\leq 1/4$ for $i\in T$.
Therefore, the $k^{th}$ term in the sum would have absolute value at most
$$
O\left(2^{-k}\left(\sum_{i\in S} m_i q_i +\sum_{i\in T}m_i (1-q_i)\right)\right) = O(2^{-k}\tilde \sigma_i).
$$
Summing over $k>\ell$, proves the appropriate bound on the error. Furthermore, 
summing this bound over $1\leq k \leq \ell$ proves that $|g_\xi| < \ell/3$, as required. 
Combining these results with the bounds on the Taylor error for $\exp'$ completes the proof.
\end{proof}

\end{document}